\documentclass[english,11pt]{article}
\usepackage{amssymb}
\usepackage{amsmath,amsfonts,amssymb,amsthm}
\usepackage{fullpage}
\usepackage{latexsym}
\usepackage{xspace}
\usepackage{babel}
\usepackage{paralist}
\usepackage{algorithm}
\usepackage{algorithmic}

\newtheorem{theorem}{Theorem}

\newtheorem{corollary}{Corollary}

\def \NP{$\mathcal{NP}$}
\def \skdf{signed $k$-dominating function}
\def \stkdf{signed total $k$-dominating function}
\def \sk{\gamma_{kS}}
\def \stk{\gamma^{t}_{kS}}
\def \usk{\Gamma_{kS}}

\begin{document}

\title{On the Signed (Total) $k$-Domination Number of a Graph\footnotemark[1]}
\author{Hongyu Liang\footnotemark[2]}

\renewcommand{\thefootnote}{\fnsymbol{footnote}}

\footnotetext[1]{This work was supported in part by the National
Basic Research Program of China Grant 2011CBA00300, 2011CBA00301,
and the National Natural Science Foundation of China Grant 61033001,
61061130540, 61073174.}

\footnotetext[2]{Institute for Interdisciplinary Information Sciences,
Tsinghua University, lianghy08@mails.tsinghua.edu.cn}

\date{}
\maketitle

\begin{abstract}
Let $k$ be a positive integer and $G=(V,E)$ be a graph of minimum degree at least $k-1$. A function $f:V\rightarrow \{-1,1\}$ is called a \emph{\skdf}~of $G$ if $\sum_{u\in N_G[v]}f(u)\geq k$ for all $v\in V$. The \emph{signed $k$-domination number} of $G$ is the minimum value of $\sum_{v\in V}f(v)$ taken over all \skdf s~of $G$. The \emph{\stkdf}~and \emph{signed total $k$-domination number} of $G$ can be similarly defined by changing the closed neighborhood $N_G[v]$ to the open neighborhood $N_G(v)$ in the definition. The \emph{upper signed $k$-domination number} is the maximum value of $\sum_{v\in V}f(v)$ taken over all \emph{minimal} \skdf s~of $G$. In this paper, we study these graph parameters from both algorithmic complexity and graph-theoretic perspectives. We prove that for every fixed $k\geq 1$, the problems of computing these three parameters are all \NP-hard. We also present sharp lower bounds on the signed $k$-domination number and signed total $k$-domination number for general graphs in terms of their minimum and maximum degrees, generalizing several known results about signed domination.
\end{abstract}

%\newcommand{\keywords}[1]{\par\addvspace\baselineskip
%\noindent\keywordname\enspace\ignorespaces#1}
%\keywords{spanning
%star forest, approximation algorithm, dense graphs}

\section{Introduction}
\label{sec:intro} All graphs considered in this paper are simple and undirected. We generally follow
\cite{gt} for standard notation and terminology in graph theory.
Let $G$ be a graph with vertex set $V(G)$ and edge set $E(G)$. The \emph{order} of $G$ is $|V(G)|$. For each vertex $v\in V(G)$, let $N_G(v)=\{u\in V(G)~|~uv\in E(G)\}$ and $N_G[v]=N_G(v)\cup \{v\}$, which are called the \emph{open neighborhood} and \emph{closed neighborhood} of $v$ (in $G$), respectively. The \emph{degree} of $v$ (in $G$) is $d_G(v)=|N_G(v)|$. The \emph{minimum degree} of $G$ is $\delta(G)=\min_{v\in V(G)}\{d_G(v)\}$, and the \emph{maximum degree} of $G$ is $\Delta(G)=\max_{v\in V(G)}\{d_G(v)\}$.
For an integer $r$, $G$ is called \emph{$r$-regular} if $\Delta(G)=\delta(G)=r$, and is called \emph{nearly $r$-regular} if $\Delta(G)=r$ and $\delta(G)=r-1$. For $S\subseteq V(G)$, $G[S]$ is the subgraph of $G$ \emph{induced} by $S$; that is, $G[S]$ is a graph with vertex set $S$ and edge set $\{uv\in E(G)~|~\{u,v\}\subseteq S\}$. For an integer $n\geq 1$, let $K_n$ denote the complete graph of order $n$; i.e., $K_n$ is an $(n-1)$-regular graph of order $n$.
For any function $f:V(G)\rightarrow \mathbb{R}$, we write $f(S)=\sum_{v\in S}f(v)$ for all $S\subseteq V(G)$, and the \emph{weight} of $f$ is $w(f)=f(V(G))$.

Domination is an important subject in graph theory, and has numerous applications in other fields; see \cite{dom_book2,dom_book} for comprehensive treatment and detailed surveys on (earlier) results in domination theory from both theoretical and applied perspectives.
A set $S\subseteq V(G)$ is called a \emph{dominating set} (resp. \emph{total dominating set}) of $G$ if
$\bigcup_{v\in S}N_G[v]=V(G)$ (resp. $\bigcup_{v\in S}N_G(v)=V(G)$). The \emph{domination number} (resp. \emph{total domination number}) of $G$, denoted
by $\gamma(G)$ (resp. $\gamma_t(G)$), is the minimum size of a dominating set (resp. total dominating set) of $G$.

Let $k\geq 1$ be a fixed integer and $G$ be a graph of minimum degree at least $k-1$. A function $f:V(G)\rightarrow \{-1,1\}$ is called a \emph{\skdf}~of $G$ if $f(N_G[v])\geq k$ for all $v\in V(G)$. The \emph{signed $k$-domination number} of $G$, denoted by $\sk(G)$, is the minimum weight of a \skdf~of $G$. When $G$ is of minimum degree at least $k$, the \emph{\stkdf}~and \emph{signed total $k$-domination number} of $G$ (denoted by $\stk(G)$) can be analogously defined by changing the closed neighborhood $N_G[v]$ to the open neighborhood $N_G(v)$ in the definition.
The concepts of signed $k$-domination number and signed total $k$-domination number are introduced in \cite{signed_k}, where sharp lower bounds of these numbers are established for general graphs, bipartite graphs and $r$-regular graphs in terms of the order of the graphs.
A related graph parameter called the \emph{upper signed $k$-domination number} of $G$, denoted by $\Gamma_{kS}(G)$, is defined in \cite{uppersk} as the maximum weight of a \emph{minimal} signed $k$-dominating function of $G$. (A signed $k$-dominating function $f$ of $G$ is called \emph{minimal} if there exists no signed $k$-dominating function $f'$ of $G$ such that $f'\neq f$ and $f'(v)\leq f(v)$ for every $v\in V(G)$.) This parameter has also been studied in \cite{uppersk_ipl}.

In the special case where $k=1$, the signed $k$-domination number and signed total $k$-domination number are exactly the \emph{signed domination number} \cite{dhhs95} and \emph{signed total domination number} \cite{z01}, respectively. These two parameters have been extensively studied in the literature; see e.g. \cite{signed_cs08,signed_cm96,dhhs95,signed_f96,signed_fm99,signed_hw02,h04,signed_m00,z01,signed_zxll99} and the references therein.

In this paper, we continue the investigation of the signed $k$-domination number and signed total $k$-domination number of graphs, from both algorithmic complexity and graph theoretic points of view. In Section \ref{sec:complexity} we show that, for every fixed $k\geq 1$, the problems of computing the signed $k$-domination number, the signed total $k$-domination number, and the upper signed $k$-domination number of a graph are all \NP-hard. We then present, in Section \ref{sec:sharp_bounds}, sharp lower bounds on the signed $k$-domination number and signed total $k$-domination number for general graphs in terms of their minimum and maximum degrees, from which several interesting results follow immediately.

\section{Complexity Issues of Signed (Total) $k$-Domination}\label{sec:complexity}

In this section we first show the \NP-hardness of computing the signed $k$-domination number and signed total $k$-domination number of a graph for all $k\geq 1$. Since the proofs for the two parameters are very similar, we only detail the proof for the signed total $k$-domination number, and merely point out the changes that need to be made for establishing hardness for the signed $k$-domination number. We now formally define the two decision problems corresponding to the computation of these two graph parameters.
\\

\textsc{Signed $k$-Domination Problem} (S$k$DP)

\textit{Instance:} A graph $G=(V,E)$ and an integer $r$.

\textit{Question:} Is $\sk(G)\leq r$?
\\

\textsc{Signed Total $k$-Domination Problem} (ST$k$DP)

\textit{Instance:} A graph $G=(V,E)$ and an integer $r$.

\textit{Question:} Is $\stk(G)\leq r$?

\begin{theorem}\label{thm:stdn}
For every integer $k\geq 1$, the ST$k$DP problem is \NP-complete.
\end{theorem}

\begin{proof}
Let $k\geq 1$ be a fixed integer. The ST$k$DP problem is clearly in \NP. We now present a polynomial-time reduction from \textsc{Minimum Total Dominating Set} (MTDS), which is a classical \NP-complete problem \cite{book_npc}, to ST$k$DP. The MTDS problem is defined as follows: Given a graph $G$ and an integer $r$, decide whether $G$ has a total dominating set of size at most $r$.

Let $(G,r)$ be an instance of the MTDS problem.
Construct another graph $H$ as follows. First let $H$ contain of a copy of $G$, which is denoted by $G'$. Also, for each vertex $v\in V(G)$, let $v'$ denote its counterpart in $G'$. For each $v\in V(G)$, we add $t(v)$ disjoint copies of $K_{k+2}$ to $H$, where $t(v)=d_{G}(v)+k-2$; call these copies $K^{v,1}_{k+2},K^{v,2}_{k+2},\ldots,K^{v,t(v)}_{k+2}$. Then, for each $i\in\{1,2,\ldots,t(v)\}$, add an edge between $v'$ and an (arbitrary) vertex from $K^{v,i}_{k+2}$. This finishes the construction of $H$. It is easy to verify that $d_H(v')=2d_G(v)+k-2$ for all $v\in V(G)$.

Let $T=(k+2)\sum_{v\in V(G)}t(v)=(k+2)\sum_{v\in V(G)}(k+d_{G}(v)-2)$ be the number of vertices in $V(H\setminus G')$.
We will prove that $\gamma_t(G)\leq r$ if and only if $\stk(H)\leq 2r-|V(G)|+T$.

First consider the ``if'' direction. Assume that $\stk(H)\leq 2r-|V(G)|+T$, and $f:V(H)\rightarrow \{-1,1\}$ is a \stkdf~of $H$ of weight $\stk(H)$.
Let $S'=\{v'\in V(G')~|~f(v')=1\}$.
It is easy to see that, for each $v\in V(G)$ and $1\leq i\leq t(v)$, all vertices in $K^{v,i}_{k+2}$ must have function value ``1'' under $f$.
%; otherwise, any vertex in $K^{v,i}_{k+2}$, except the one that is adjacent to $v'$, has neighborhood sum less than $k$ under $f$, which contradicts with the property of a \stkdf. Let $S'=\{v'\in V(H)~|~f(v')=-1\}$.
It follows that $\stk(H)=w(f)=T+|S'|-(|V(G')|-|S'|)=2|S'|-|V(G)|+T$. Since $\stk(H)\leq 2r-|V(G)|+T$, we have $|S'|\leq r$. Now define $S=\{v\in V(G)~|~v'\in S'\}$; i.e., $S$ is the counterpart of $S'$ in $G$. We show that $S$ is a total dominating set of $G$. Assume to the contrary that $S$ is not a total dominating set of $G$, and let $v\in V(G)$ be such that $N_G(v)\cap S=\emptyset$. By our definitions of $S$ and $S'$, $f(u')=-1$ for all $u\in N_G(v)$. Thus, $\sum_{x\in N_{H}(v')}f(x)\leq t(v)-d_G(v)=k-2$, contradicting with the fact that $f$ is a \stkdf~of $H$. Therefore, $S'$ is indeed a total dominating set of $G$, from which $\gamma_t(G)\leq |S'|\leq r$ follows. This completes the proof for the ``if'' direction.

Now comes the ``only if'' part of the reduction. Suppose $\gamma_t(G)\leq r$ and $S\subseteq V(G)$ is a total dominating set of $G$ of size at most $r$. Define a function $f:V(H)\rightarrow \{-1,1\}$ as follows: $f(x)=-1$ if $x=v'$ for some $v\in V(G)\setminus S$, and $f(x)=1$ otherwise. The weight of $f$ is $T+|S|-(|V(G)|-|S|)=2|S|-|V(G)|+T\leq 2r-|V(G)|+T$. We now verify that $f$ is a \stkdf~of $H$. For each $x\in V(H\setminus G')$, $f(N_{H}(x))\geq (k+1)-1=k$. For each $v'\in V(G')$ (with $v\in V(G)$), since $S$ is a total dominating set of $G$, $f(N_{H}(v'))\geq t(v)+1-(d_G(v)-1)=t(v)+2-d_G(v)=k$. Hence, $f$ is a \stkdf~of $H$ of weight at most $2r-|V(G)|+T$. This completes the ``only if'' part of the reduction.

Therefore, $\gamma_t(G)\leq r$ if and only if $\stk(H)\leq 2r-|V(G)|+T$. This finishes the whole reduction, and hence concludes the proof of Theorem~\ref{thm:stdn}.
\end{proof}

\begin{theorem}\label{thm:sdn}
For every integer $k\geq 1$, the S$k$DP problem is \NP-complete.
\end{theorem}
\begin{proof}
The proof is very similar to that of Theorem~\ref{thm:stdn}, with two differences in the reduction. Therefore, we only describe the reduction.  We reduce from the \NP-complete problem \textsc{Minimum Dominating Set} (which, given a graph $G$ and an integer $r$, needs to decide whether $G$ has a dominating set of size at most $r$) to S$k$DP. Let $(G,r)$ be an instance of \textsc{Minimum Dominating Set}.
Construct another graph $H$ as follows. First let $H$ contain of a copy of $G$, which is denoted by $G'$. For each vertex $v\in V(G)$, add $s(v)$ disjoint copies of $K_{k+1}$ to $H$, where $s(v)=d_{G}(v)+k-1$; call these copies $K^{v,1}_{k+1},K^{v,2}_{k+1},\ldots,K^{v,s(v)}_{k+1}$. Then, for each $i\in\{1,2,\ldots,s(v)\}$, add an edge between $v'$ (the counterpart of $v$ in $G'$) and an arbitrary vertex from $K^{v,i}_{k+1}$. This finishes the construction of $H$. Using similar argument to that in Theorem~\ref{thm:stdn}, we can prove that $\gamma(G)\leq r$ if and only if $\sk(H)\leq 2r-|V(G)|+T$, where $T=(k+1)\sum_{v\in V(G)}s(v)$. The \NP-completeness of S$k$DP is thus established.
\end{proof}

We now define the problem corresponding to the computation of the upper signed $k$-domination number of graphs as follows.
\\

\textsc{Upper Signed $k$-Domination Problem} (US$k$DP)

\textit{Instance:} A graph $G=(V,E)$ and an integer $r$.

\textit{Question:} Is $\usk(G)\geq r$?

\begin{theorem}
For every integer $k\geq 1$, the US$k$DP problem is \NP-complete.
\end{theorem}

\begin{proof}
The US$k$DP problem is in \NP~because given a function $f: V(G)\rightarrow \{-1,1\}$, we can verify in polynomial time whether $f$ is a minimal signed $k$-dominating function of $G$ using Lemma 4 in \cite{uppersk_ipl}. We will describe a polynomial time reduction from the 1-in-3 SAT problem to it. The 1-in-3 SAT problem is defined as follows: Given a Boolean formula in conjunctive normal form, each clause of which contains exactly three \emph{positive} literals (i.e., variables with no negations), decide whether the formula is \emph{1-in-3 satisfiable}, i.e., if there exists an assignment of the variables such that exactly one variable of each clause is assigned TRUE. This problem is known to be \NP-complete \cite{sat_s78}.

Let $F$ be a Boolean formula with variables $\{x_1,x_2,\ldots,x_n\}$, which is an input of the 1-in-3 SAT problem. Assume $F=\bigwedge_{i=1}^{m}c_i$ where $c_i=(x_{i_1}\lor x_{i_2} \lor x_{i_3})$ for each $i\in \{1,2,\ldots,m\}$.
We construct a graph $G$ as follows. Take $m$ disjoint copies of $K_{k+2}$, each of which corresponds to a clause $c_i$ with $i\in\{1,2,\ldots,m\}$, and $n$ disjoint copies of $K_{k+3}$ (also disjoint from the copies of $K_{k+2}$'s) each of which corresponds to a variable $x_j$ with $j\in\{1,2,\ldots,n\}$. Delete one edge from each copy of $K_{k+3}$. We will call the copy of $K_{k+2}$ corresponding to $c_i$ the \emph{$i$-th clause block}, and call the copy of $K_{k+3}$ (with one edge missing) corresponding to $x_j$ the \emph{$j$-th variable block}. For each $i\in\{1,2,\ldots,m\}$, let $c'_i$ be an (arbitrary) vertex in the $i$-th clause block. For every $j\in\{1,2,\ldots,n\}$, let $x'_j$ and $x''_j$ be the two vertices in the $j$-th variable block for which the edge $x'_jx''_j$ is removed. For each clause $c_i=(x_{i_1}\lor x_{i_2} \lor x_{i_3})$, add three cross-block edges $c'_ix'_{i_1}, c'_ix'_{i_2}$, and $c'_ix'_{i_3}$. This finishes the construction of $G$. Note that $|V(G)|=(k+3)n+(k+2)m$.

We claim that $\usk(G)\geq (k+1)n+(k+2)m$ if and only if $F$ is 1-in-3 satisfiable. First consider the ``if'' direction, and let $\mathcal{A}:\{x_1,x_2,\ldots,x_n\}\rightarrow \{\textrm{TRUE, FALSE}\}$ be an assignment that witnesses the 1-in-3 satisfiability of $F$. Define $f:V(G)\rightarrow \{-1,1\}$ as follows: For each $j\in\{1,2,\ldots,n\}$, let
\begin{eqnarray*}
f(x'_j)=\left\{
\begin{array}{ll}
1 & \textrm{~if~}\mathcal{A}(x_j)=\textrm{TRUE};
\\
-1 & \textrm{~if~}\mathcal{A}(x_j)=\textrm{FALSE}
\end{array}
\right.
\textrm{~~and~~}
f(x''_j)=\left\{
\begin{array}{ll}
-1 & \textrm{~if~}\mathcal{A}(x_j)=\textrm{TRUE};
\\
1 & \textrm{~if~}\mathcal{A}(x_j)=\textrm{FALSE}.
\end{array}
\right.
\end{eqnarray*}
Let $f(v)=1$ for all $v\in V(G)\setminus \bigcup_{j=1}^{n}\{x'_j,x''_j\}$.

Clearly, $w(f)=(k+1)n+(k+2)m$. Since exactly one of $\mathcal{A}(x_{i_1}), \mathcal{A}(x_{i_2})$ and $\mathcal{A}(x_{i_3})$ is TRUE for each $1\leq i\leq m$, it is easy to verify that $f$ is a signed $k$-dominating function of $G$. We next prove that $f$ is minimal, that is, for every vertex $v\in V(G)$ with $f(v)=1$ there exists $u\in N_G[v]$ for which $f(N_G[u])\in\{k,k+1\}$ (see \cite{uppersk_ipl}). For every $j\in\{1,2,\ldots,n\}$, there is (at least) one vertex $u$ in the $j$-th variable block such that $u\not\in\{x'_j,x''_j\}$. This vertex $u$ is adjacent to all other vertices in the $j$-th variable block, and clearly $f(N_G[u])=k+1$. For every $i\in\{1,2,\ldots,m\}$, $c'_i$ is adjacent to all other vertices in the $i$-th clause block, and $f(N_G[c'_i])=(k+2)+(1-2)=k+1$ since exactly one of $f(x'_{i_1}), f(x'_{i_2})$ and $f(x'_{i_3})$ is 1. Therefore, $f$ is indeed a minimal \skdf~of $G$ with weight $(k+1)n+(k+2)m$, and the correctness of the ``if'' direction follows.

We now turn to the ``only if'' part of the claim. Assume that $f$ is a minimal \skdf~of $G$ of weight at least $(k+1)n+(k+2)m$. If for some $j\in\{1,2,\ldots,n\}$, the vertices in the $j$-th variable block all have value 1 under $f$, then $f(N_G[v])\geq k+2$ for every $v\neq x'_j$ in the $j$-th variable block. Thus, there is no $u\in N_G[x''_j]$ such that $f(N_G[u])\in \{k,k+1\}$, which violates the minimality of $f$. Hence, at least one vertex from each variable block must have value $-1$ under $f$, implying that $w(f)\leq (k+1)n+(k+2)m$. We thus have $w(f)=(k+1)n+(k+2)m$, and therefore (1) $f(v)=1$ for every vertex $v$ in the clause blocks, and (2) for each $j\in\{1,2,\ldots,n\}$, $f(v)=-1$ for exactly one vertex $v$ in the $j$-th variable block. Now produce an assignment $\mathcal{A}$ as follows: For each $j\in\{1,2,\ldots,n\}$, let $\mathcal{A}(x_j)=$TRUE if $f(x'_j)=1$, and $\mathcal{A}(x_j)=$FALSE otherwise. For every $i\in\{1,2,\ldots,m\}$, we have $k\leq f(N_G[c'_i])=(k+2)+f(x_{i_1})+f(x_{i_2})+f(x_{i_3})$, and thus at least one of $f(x_{i_1}), f(x_{i_2})$ and $f(x_{i_3})$ must be 1. Assume that at least two of the three values are 1. Then $f(N_G[c'_i])\geq k+3$, and obviously $f(N_G[v])=k+2$ for every other vertex $v$ in the $i$-th clause block. This indicates, however, that a vertex $v\neq c'_i$ in the $i$-th clause block does not have any neighbor (including itself) whose closed-neighborhood-sum is $k$ or $k+1$, contradicting with the minimality of $f$. Accordingly, exactly one of $f(x_{i_1}), f(x_{i_2})$ and $f(x_{i_3})$ is 1, and thus exactly one of $\mathcal{A}(x_{i_1}), \mathcal{A}(x_{i_2})$ and $\mathcal{A}(x_{i_3})$ is TRUE, for every $i\in\{1,2,\ldots,n\}$. Therefore, $F$ is 1-in-3 satisfiable, finishing the proof of the ``only if'' part of the reduction.

The reduction is completed and the \NP-completeness of US$k$DP is thus established.
\end{proof}

\section{Sharp Lower Bounds on $\sk(G)$ and $\stk(G)$}\label{sec:sharp_bounds}

In this section we present sharp lower bounds on $\sk(G)$ and $\stk(G)$ in terms of the minimum and maximum degrees of $G$. Let $k\geq 1$ be a fixed integer throughout this section. For each integer $n$, define $I_{n}=1$ if $n\equiv k \pmod 2$, and $I_{n}=0$ otherwise; that is, $I_{n}$ is the indicator variable of whether $n$ and $k$ have the same parity.

\begin{theorem}\label{thm:lb_sk}
For every graph $G$ with $\delta(G)\geq k-1$,
$$\sk(G) \geq |V(G)|\cdot\frac{\delta(G)-\Delta(G)+2k+I_{\delta(G)}+I_{\Delta(G)}}{\delta(G)+\Delta(G)+2+I_{\delta(G)}-I_{\Delta(G)}}\;.$$
%and the inequality is strict if $\delta(G)\leq k$ and $\Delta(G)\geq k+1$.
\end{theorem}

\begin{proof}
Let $G$ be a graph of order $n$ with $\delta(G)\geq k-1$. For notational simplicity, we write $\delta$ and $\Delta$ to denote $\delta(G)$ and $\Delta(G)$ respectively. When $\delta=\Delta$, it is easy to verify that the theorem degenerates to Theorem 5 in \cite{signed_k}. Thus, we assume in what follows that $\Delta\geq \delta+1$. Let $f$ be a \skdf~of $G$ of weight $\sk(G)$. We need to introduce some notations. Let $P=\{v\in V(G)~|~f(v)=1\}$ and $Q=V(G)\setminus P=\{v\in V(G)~|~f(v)=-1\}$. Furthermore, denote $P_{\delta}=\{v\in P~|~d_G(v)=\delta\}$, $P_{\Delta}=\{v\in P~|~d_G(v)=\Delta\}$, and $P_{m}=P\setminus (P_{\delta}\cup P_{\Delta})$. Define $Q_{\delta}$, $Q_{\Delta}$, and $Q_{m}$ analogously. For each $c\in\{\delta,\Delta,m\}$, let $V_{c}=P_{c}\cup Q_{c}$. Notice that $V_{\delta}\cap V_{\Delta}=\emptyset$ since $\Delta>\delta$. Let $R=\{v\in V(G)~|~d_G(v)\equiv k\pmod 2\}$.
Clearly $\sum_{y\in N_G[x]}f(y)\geq k+1$ for each $x\in R$. Thus, we have
\begin{eqnarray*}\label{equ:1}
kn+|R|&\leq& \sum_{x\in V(G)}\sum_{y\in N_G[x]}f(y)=\sum_{x\in V(G)}(d_G(x)+1)f(x)\\
&=&(\delta+1)|P_{\delta}|+(\Delta+1)|P_{\Delta}|+\sum_{x\in P_{m}}(d_G(x)+1)
-(\delta+1)|Q_{\delta}|-(\Delta+1)|Q_{\Delta}|-\sum_{x\in Q_m}(d_G(x)+1)\\
&\leq&(\delta+1)|P_{\delta}|+(\Delta+1)|P_{\Delta}|+\Delta |P_m|-(\delta+1)|Q_{\delta}|-(\Delta+1)|Q_{\Delta}|-(\delta+2)|Q_m|\\
& &\textrm{(since}~\delta+1\leq d_G(x)\leq \Delta-1~\textrm{for each}~x\in P_m\cup Q_m)\\
&=&(\delta+1)|V_{\delta}|+(\Delta+1)|V_{\Delta}|+\Delta |V_m|-2(\delta+1)|Q_{\delta}|-2(\Delta+1)|Q_{\Delta}|-(\Delta+\delta+2)|Q_m|\\
&=&(\Delta+1)n-(\Delta-\delta)|V_{\delta}|-|V_m|-(\Delta+\delta+2)|Q|+(\Delta-\delta)|Q_{\delta}|-(\Delta-\delta)|Q_{\Delta}|\\
& &\textrm{(note that~}n=|V(G)|=|V_{\delta}|+|V_{\Delta}|+|V_{m}|\textrm{~and~} |Q|=|Q_{\delta}|+|Q_{\Delta}|+|Q_{m}|).
\end{eqnarray*}
Therefore,
\begin{eqnarray*}
(\Delta+1-k)n &\geq& |R|+|V_m|+(\Delta-\delta)(|V_{\delta}|-|Q_{\delta}|+|Q_{\Delta}|)+(\Delta+\delta+2)|Q|\\
&=& |R|+|V_m|+(\Delta-\delta)(|P_{\delta}|+|Q_{\Delta}|)+(\Delta+\delta+2)|Q|.
\end{eqnarray*}

Since $R=\{v\in V(G)~|~d(v)\equiv k\pmod 2\}$, it holds that $V_{\delta}\subseteq R$ if $\delta\equiv k \pmod 2$, and that $V_{\Delta}\subseteq R$ if $\Delta\equiv k \pmod 2$. Recalling that $V_\Delta \cap V_\delta=\emptyset$, we have $|R|\geq I_{\delta}\cdot |V_{\delta}|+I_\Delta\cdot |V_{\Delta}|$. Thus,
\begin{eqnarray*}
(\Delta+1-k)n &\geq& I_{\delta}\cdot |V_{\delta}|+I_\Delta\cdot |V_{\Delta}|+|V_m|+(\Delta-\delta)(|P_\delta|+|Q_{\Delta}|)+(\Delta+\delta+2)|Q|\\
&=&I_{\Delta}(|V_m|+|V_\delta|+|V_\Delta|)+(1-I_{\Delta})|V_m|+(I_\delta-I_\Delta)|V_{\delta}|\\
& &+(\Delta-\delta)(|P_\delta|+|Q_{\Delta}|)+(\Delta+\delta+2)|Q|\\
&=&I_{\Delta}\cdot n+(1-I_{\Delta})|V_m|+(I_\delta-I_\Delta)|V_{\delta}|
+(\Delta-\delta)(|P_\delta|+|Q_{\Delta}|)+(\Delta+\delta+2)|Q|.
\end{eqnarray*}

Observing that $\Delta-\delta\geq 1\geq \max\{I_\delta-I_\Delta,I_\Delta-I_\delta\}$ and
$(1-I_\Delta)|V_m|\geq (1-I_\Delta)|Q_m|\geq (I_\delta-I_\Delta)|Q_m|$,
we get
\begin{eqnarray*}
& &(\Delta+1-k-I_\Delta)n \nonumber \\
&\geq&
(I_\delta-I_{\Delta})|Q_m|+(I_\delta-I_\Delta)|V_{\delta}|
+(I_\Delta-I_\delta)|P_\delta|+(I_\delta-I_\Delta)|Q_{\Delta}|+(\Delta+\delta+2)|Q| \label{equ:4}\\
&=&(I_\delta-I_{\Delta})(|Q_m|+|V_{\delta}|-|P_{\delta}|+|Q_{\Delta}|)+(\Delta+\delta+2)|Q|\nonumber\\
&=&(I_{\delta}-I_{\Delta})(|Q_m|+|Q_\delta|+|Q_{\Delta}|)+(\Delta+\delta+2)|Q| \nonumber \\
&=&(\Delta+\delta+2+I_{\delta}-I_{\Delta})|Q| \nonumber .
\end{eqnarray*}
Hence, we deduce that
\begin{eqnarray*}
|Q|\leq n\cdot \frac{\Delta-k+1-I_\Delta}{\delta+\Delta+2+I_{\delta}-I_{\Delta}}\;,
\end{eqnarray*}
from which it follows that
\begin{eqnarray*}\label{equ:5}
\sk(G)=n-2|Q|\geq n\cdot \frac{\delta-\Delta+2k+I_\delta+I_\Delta}{\delta+\Delta+2+I_{\delta}-I_{\Delta}}\;,
\end{eqnarray*}
which is exactly the desired inequality in Theorem~\ref{thm:lb_sk}.
%Now assume that $\delta\leq k$ and $\Delta\geq k+1$. We wish to prove that the inequality is strict in this case.
%Assume to the contrary that (\ref{equ:5}) becomes equality for some $G$. By (\ref{equ:4}) it must both that
%$(\Delta-\delta)|P_\delta|=(I_\Delta-I_\delta)|P_\delta|$.
%If there is some $v\in Q_\delta$, then $f[v]\leq \delta+1-2=\delta-1<k$, contradicting with the fact that $f$ is a \skdf~of $G$. Hence, $Q_\delta=\emptyset$, which indicates that $|P_\delta|>0$ (since $G$ must contain some vertex of degree $\delta$). As $\Delta-\delta\geq 1$ and $(\Delta-\delta)|P_\delta|=(I_\Delta-I_\delta)|P_\delta|$, we have $\Delta-\delta=1$, $I_\Delta=1$, and $I_\delta=0$. Thus, $\delta$ must be $k-1$ and $\Delta=k$. This, however, contradicts with the assumption that $\Delta\geq k+1$. Hence, the inequality in (\ref{equ:5}) is strict. This completes the proof of Theorem~\ref{thm:lb_sk}.
\end{proof}

A vertex of degree $k-1$ or $k$ in a graph $G$ clearly has function value 1 under all \skdf s~of $G$. Thus, it is natural to consider graphs with minimum degree at least $k+1$ (as is done in \cite{uppersk_ipl} for establishing sharp upper bounds for the upper signed $k$-domination number).
We next show that Theorem~\ref{thm:lb_sk} is sharp for all $\Delta\geq \delta\geq k+1$. This level of sharpness is high as it applies not only to special values of minimum and maximum degrees.

\begin{theorem}\label{thm:sk_sharp}
For any integers $\delta$ and $\Delta$ such that $\Delta\geq \delta\geq k+1$,
%or $k\geq \Delta\geq \delta\geq k-1$,
there exists an infinite family $\mathcal{F}$ of graphs with minimum degree $\delta$ and maximum degree $\Delta$, such that for every graph $G\in \mathcal{F}$,
$$\sk(G) = |V(G)|\cdot\frac{\delta-\Delta+2k+I_{\delta}+I_{\Delta}}{\delta+\Delta+2+I_{\delta}-I_{\Delta}}\;.$$
%For any integers $\delta\in\{k-1,k\}$ and $\Delta\geq k+1$, and any $\epsilon>0$, there exists a graph $G$ with maximum degree $\Delta$ and minimum degree $\delta$ such that
%$$|V(G)|\cdot \frac{\delta-\Delta+2k+I_{\delta}+I_{\Delta}}{\delta+\Delta+2+I_{\delta}-I_{\Delta}} < \sk(G) < |V(G)|\cdot\left(\frac{\delta-\Delta+2k+I_{\delta}+I_{\Delta}}{\delta+\Delta+2+I_{\delta}-I_{\Delta}}+\epsilon\right)\;.$$
\end{theorem}

\begin{proof}
Fix integers $\Delta$ and $\delta$ such that $\Delta\geq \delta\geq k+1$. %We investigate two cases separately.
%
%\textbf{Case 1:} $k-1\leq \delta\leq \Delta\leq k$.
%
%It is easy to verify that $T=1$ in this case. Let $G$ be any graph whose vertices all have degree $k-1$ or $k$. Then, it is clear that any vertex must have function value 1 under any \skdf~of $G$, and thus $\sk(G)=|V(G)|=T\cdot |V(G)|$.
%
%\textbf{Case 2:} $\Delta\geq \delta \geq k+1$.
Let $H_1,H_2,\ldots,H_{t}$ be $t$ disjoint copies of the complete bipartite graph $K_{a,b}$ with vertex partition $(A,B)$, where $|A|=a=(\delta+k+1+I_\delta)/2$, $|B|=b=(\Delta-k+1-I_\Delta)/2$ (it is easy to verify that $a$ and $b$ are both integers), and $t$ is an arbitrary \emph{even} integer larger than $\Delta$. It is also easy to check that $1\leq a\leq \delta$ and $1\leq b\leq \Delta$ (just note that $I_\delta=0$ when $\delta=k+1$). For each $1\leq i\leq t$, let $A_i$ and $B_i$ denote the vertex partition of $H_i$ with size $a$ and $b$, respectively. Let $P=\bigcup_{i=1}^{t}A_i$ and $Q=\bigcup_{i=1}^{t}B_i$. Note that each vertex in $P$ is connected to exactly $b$ vertices in $Q$, and each vertex in $Q$ is adjacent to exactly $a$ vertices in $P$.

Our desired graph $G$ has vertex set $P\cup Q$, and contains $\bigcup_{i=1}^{t}H_i$ as a subgraph. Furthermore, we add some edges between vertices in $P$ to make $G[P]$ become $(\Delta-b)$-regular (no edges need to be added if $\Delta=b$). This can be done in the following way: Imagine that there is a complete graph $K$ whose vertex set is $P$. Since $|P|=ta$ is even and every complete graph of even order is 1-factorable (see e.g. Theorem 9.1 in \cite{gt_harary}), the edges of $K$ can be partitioned into $|P|-1\geq \Delta$ perfect matchings of $K$. Taking $\Delta-b$ of these matchings and adding them to $G$ certainly makes $G[P]$ become $(\Delta-b)$-regular. Similarly, we add some edges between vertices in $Q$ to make $G[Q]$ $(\delta-a)$-regular. This finishes the construction of $G$. Note that all vertices in $P$ have degree $\Delta$ and those in $Q$ have degree $\delta$, and thus $G$ is of minimum degree $\delta$ and maximum degree $\Delta$. (Note also that by varying $t$, we get an infinite family of graphs with the desired properties.)

Define a function $f:P\cup Q\rightarrow \{-1,1\}$ by letting $f(v)=1$ for all $v\in P$ and $f(u)=-1$ for all $u\in Q$. Then, for each $v\in P$, $f(N_G[v])=\Delta+1-2b=k+I_\Delta\geq k$, and for each $u\in Q$, $f(N_G[u])=2a-(\delta+1)=k+I_\delta\geq k$. Therefore, $f$ is a \skdf~of $G$. Since $|V(G)|=|P|+|Q|$ and $|P|/|Q|=a/b=\frac{\delta+k+1+I_\delta}{\Delta-k+1-I_\Delta}$, we have
\begin{eqnarray*}
\sk(G)\leq w(f)=|P|-|Q|=(1-\frac{2}{|P|/|Q|+1})|V(G)|=|V(G)|\cdot\frac{\delta-\Delta+2k+I_{\delta}+I_{\Delta}}{\delta+\Delta+2+I_{\delta}-I_{\Delta}}\;.
\end{eqnarray*}
By Theorem~\ref{thm:lb_sk}, we know that the equality holds in the above formula, which completes the proof of Theorem~\ref{thm:sk_sharp}.
%\textbf{Case 3:} $\Delta\geq k+1$ and $k-1\leq \delta\leq k$.
\end{proof}

We can also derive a sharp lower bound on the signed total $k$-domination number of a graph as follows.
\begin{theorem}\label{thm:lb_stk}
For every graph $G$ with $\delta(G)\geq k$,
$$\stk(G) \geq |V(G)|\cdot\frac{\delta(G)-\Delta(G)+2k+2-I_{\delta(G)}-I_{\Delta(G)}}{\delta(G)+\Delta(G)+I_{\Delta(G)}-I_{\delta(G)}}\;.$$
%and the inequality is strict if $\delta(G)\leq k$ and $\Delta(G)\geq k+1$.
\end{theorem}

\begin{theorem}\label{thm:stk_sharp}
For any integers $\delta$ and $\Delta$ such that $\Delta\geq \delta\geq k+2$,
%or $k\geq \Delta\geq \delta\geq k-1$,
there exists an infinite family $\mathcal{F}$ of graphs with minimum degree $\delta$ and maximum degree $\Delta$, such that for every graph $G\in \mathcal{F}$,
$$\stk(G) = |V(G)|\cdot\frac{\delta-\Delta+2k+2-I_{\delta}-I_{\Delta}}{\delta+\Delta+I_{\Delta}-I_{\delta}}\;.$$
\end{theorem}

The proofs of Theorems~\ref{thm:lb_stk} and \ref{thm:stk_sharp} are very similar to those of Theorems~\ref{thm:lb_sk} and
\ref{thm:sk_sharp}, and thus are put in the appendix.

Theorems~\ref{thm:lb_sk} and \ref{thm:lb_stk} are generalizations of Theorem 5 in \cite{signed_k}.
The following corollaries, which generalize some other known results regarding signed domination number and signed total domination number, are also immediate from the preceding theorems.

\begin{corollary}
For any nearly $r$-regular graph $G$ of order $n$ with $r\geq k$, $\sk(G)\geq kn/(r+I_{r-1})$ and $\stk(G)\geq kn/(r-I_{r-1})$.
\end{corollary}

\begin{corollary}
Let $c$ be a real number for which $-1< c\leq 1$. Then $\sk(G)\geq cn$ for every graph $G$ of order $n$ with $\delta(G)\geq k-1$ and $\Delta(G)\leq ((1-c)\delta(G)+2k-2c)/(1+c)$, and $\stk(G)\geq cn$ for every graph $G$ of order $n$ with $\delta(G)\geq k$ and $\Delta(G)\leq ((1-c)\delta(G)+2k)/(1+c)$.
\end{corollary}

\begin{corollary}
Let $G$ be a graph with $\delta(G)\geq k$ and $\Delta(G)\leq \delta(G)+2k$. Then $\sk(G)\geq 0$ and $\stk(G)\geq 0$.
\end{corollary}

\bibliographystyle{plain}
\bibliography{signed_k}

\appendix
\section{Proof of Theorem~\ref{thm:lb_stk}}\label{apx:pf}
\begin{proof}[Proof of Theorem~\ref{thm:lb_stk}]
Let $G$ be a graph of order $n$ and $f$ be a \stkdf~of $G$. Let $\delta,\Delta,P,Q,P_\delta,P_\Delta,P_m,Q_\delta,Q_\Delta,Q_m,V_\delta,V_\Delta,V_m$ be defined in the same way as in the proof of Theorem~\ref{thm:lb_sk}. Let $R=\{v\in V(G)~|~d(v)\not\equiv k\pmod 2\}$ (which is different from the definition of $R$ in the proof of Theorem~\ref{thm:lb_sk}). Assume $\Delta>\delta$, otherwise the theorem just becomes Theorem 5 in \cite{signed_k}. Since $\sum_{y\in N_G(x)}f(y)\geq k+1$ for all $x\in R$, we have:
\begin{eqnarray*}
kn+|R|&\leq& \sum_{x\in V(G)}\sum_{y\in N_G(x)}f(y)\\
&=&\sum_{x\in V(G)}d_G(x)f(x)\\
&=&\delta|P_{\delta}|+\Delta|P_{\Delta}|+\sum_{x\in P_{m}}d_G(x)
-\delta|Q_{\delta}|-\Delta|Q_{\Delta}|-\sum_{x\in Q_m}d_G(x)\\
&\leq&\delta|P_{\delta}|+\Delta|P_{\Delta}|+(\Delta-1) |P_m|-\delta|Q_{\delta}|-\Delta|Q_{\Delta}|-(\delta+1)|Q_m|\\
&=&\delta|V_{\delta}|+\Delta|V_{\Delta}|+(\Delta-1) |V_m|-2\delta|Q_{\delta}|-2\Delta|Q_{\Delta}|-(\Delta+\delta)|Q_m|\\
&=&\Delta n-(\Delta-\delta)|V_{\delta}|-|V_m|-(\Delta+\delta)|Q|+(\Delta-\delta)|Q_{\delta}|-(\Delta-\delta)|Q_{\Delta}|\\
& &\textrm{(recall that~}n=|V(G)|=|V_{\delta}|+|V_{\Delta}|+|V_{m}|\textrm{~and~} |Q|=|Q_{\delta}|+|Q_{\Delta}|+|Q_{m}|).
\end{eqnarray*}
By our definition, it holds that $|R|\geq (1-I_\delta)|V_\delta|+(1-I_\Delta)|V_\Delta|$. Therefore,
\begin{eqnarray*}
(\Delta-k)n &\geq& |R|+|V_m|+(\Delta-\delta)(|V_{\delta}|-|Q_{\delta}|+|Q_{\Delta}|)+(\Delta+\delta)|Q|\\
&=& |R|+|V_m|+(\Delta-\delta)(|P_{\delta}|+|Q_{\Delta}|)+(\Delta+\delta)|Q|\\
&\geq& (1-I_{\delta})|V_{\delta}|+(1-I_\Delta)|V_{\Delta}|+|V_m|+(\Delta-\delta)(|P_\delta|+|Q_{\Delta}|)+(\Delta+\delta)|Q|\\
&=&(1-I_{\Delta})(|V_m|+|V_\delta|+|V_\Delta|)+I_{\Delta}|V_m|+(I_\Delta-I_\delta)|V_{\delta}|\\
& &+(\Delta-\delta)(|P_\delta|+|Q_{\Delta}|)+(\Delta+\delta)|Q|\\
&=&(1-I_{\Delta})n+I_{\Delta}|V_m|+(I_\Delta-I_\delta)|V_{\delta}|
+(\Delta-\delta)(|P_\delta|+|Q_{\Delta}|)+(\Delta+\delta)|Q|.
\end{eqnarray*}
Noting that $I_{\Delta}|V_m|\geq (I_\Delta-I_\delta)|Q_m|$ and $\Delta-\delta\geq \max\{I_\Delta-I_\delta,I_\delta-I_\Delta\}$, we obtain
\begin{eqnarray*}
& &(\Delta-k+I_\Delta-1)n \\
&\geq&
I_{\Delta}|V_m|+(I_\Delta-I_\delta)|V_{\delta}|
+(\Delta-\delta)(|P_\delta|+|Q_{\Delta}|)+(\Delta+\delta)|Q|\\
&\geq&
(I_\Delta-I_{\delta})|Q_m|+(I_\Delta-I_\delta)|V_{\delta}|
+(I_\delta-I_\Delta)|P_\delta|+(I_\Delta-I_\delta)|Q_{\Delta}|+(\Delta+\delta)|Q|\\
&=&(I_\Delta-I_{\delta})(|Q_m|+|V_\delta|-|P_\delta|+|Q_\Delta|)+(\Delta+\delta)|Q|\\
&=&(I_{\Delta}-I_{\delta})|Q|+(\Delta+\delta)|Q|\\
&=&(\Delta+\delta+I_{\Delta}-I_{\delta})|Q|.
\end{eqnarray*}
Hence, we have
\begin{eqnarray*}
|Q|\leq n\cdot \frac{\Delta-k+I_\Delta-1}{\delta+\Delta+I_{\Delta}-I_{\delta}}\;,
\end{eqnarray*}
from which it follows that
\begin{eqnarray*}\label{equ:5}
\sk(G)=n-2|Q|\geq n\cdot \frac{\delta-\Delta+2k+2-I_\delta-I_\Delta}{\delta+\Delta+I_{\Delta}-I_{\delta}}\;,
\end{eqnarray*}
completing the proof of Theorem~\ref{thm:lb_stk}.
\end{proof}

\section{Proof of Theorem~\ref{thm:stk_sharp}}\label{apx:pf2}
\begin{proof}[Proof of Theorem~\ref{thm:stk_sharp}]
Fix integers $\Delta$ and $\delta$ such that $\Delta\geq \delta\geq k+2$. %We investigate two cases separately.
%
%\textbf{Case 1:} $k-1\leq \delta\leq \Delta\leq k$.
%
%It is easy to verify that $T=1$ in this case. Let $G$ be any graph whose vertices all have degree $k-1$ or $k$. Then, it is clear that any vertex must have function value 1 under any \skdf~of $G$, and thus $\sk(G)=|V(G)|=T\cdot |V(G)|$.
%
%\textbf{Case 2:} $\Delta\geq \delta \geq k+1$.
We proceed with the same construction used in the proof of Theorem~\ref{thm:sk_sharp}, except for setting $a=(\delta+k-I_\delta+1)/2$ and $b=(\Delta-k+I_\Delta-1)/2$ instead. (It is easy to check that $a$ and $b$ are integers satisfying that $1\leq a\leq \delta$ and $1\leq b\leq \Delta$.) The obtained graph $G$ has vertex set $P\cup Q$, where $d_G(v)=\Delta$ for all $v\in P$ and $d_G(u)=\delta$ for all $u\in Q$. Furthermore, each vertex $v\in P$ is adjacent to exactly $b$ vertices in $Q$ and $\Delta-b$ vertices in $P$, while every vertex $u\in Q$ is adjacent to precisely $a$ vertices in $P$ and $\delta-a$ vertices in $Q$. Now define a function $f$ which assigns 1 to all vertices in $P$ and $-1$ to those in $Q$. It is easy to verify that $f$ is a \stkdf~of $G$ with weight $|V(G)|\cdot \frac{\delta-\Delta+2k+2-I_\delta-I_\Delta}{\delta+\Delta+I_{\Delta}-I_{\delta}}$, completing the proof of Theorem~\ref{thm:stk_sharp}.
\end{proof}

\end{document}